\documentclass[12pt,a4paper]{article}
\usepackage[a4paper,margin=1.25in,footskip=0.25in]{geometry}
\usepackage{verbatim}
\usepackage{amsthm}
\usepackage{amsmath}
\usepackage{amssymb}
\usepackage{caption}
\usepackage[dvipdfmx,unicode,bookmarks=true, bookmarksnumbered=true, bookmarksopen = false, colorlinks=true]{hyperref}
\usepackage{graphics}
\usepackage{epsfig}
\usepackage{geometry}
\makeatletter \renewenvironment{proof}[1][\proofname]
{\par\pushQED{\qed}\normalfont\topsep6\p@\@plus6\p@\relax\trivlist\item[\hskip\labelsep\bfseries#1\@addpunct{.}]\ignorespaces}{\popQED\endtrivlist\@endpefalse} \makeatother

\theoremstyle{plain}
\newtheorem{thm}{Theorem}
\newtheorem{lem}{Lemma}

\newtheorem{prop}{Proposition}

\begin{document}
\author{Yangbo Song\thanks{School of Management and Economics, the Chinese University of Hong Kong (Shenzhen). Email: yangbosong@cuhk.edu.cn.}
\and
Mofei Zhao\thanks{Graduate School at Shenzhen, Tsinghua University. Email: 	ellipsis010011@sina.com.}}
\date{\today}
\title{Dynamic R\&D Competition under Uncertainty and Strategic Disclosure}
\maketitle

\begin{abstract}

We study a model of dynamic two-stage R\&D competition where the competing firms are uncertain about the difficulty of the first stage. Staying in the competition is costly and a firm can also choose whether and when to quit. When a firm solves the first stage, it can choose whether and when to disclose the solution. We find that there exists a unique symmetric equilibrium, in which each firm will disclose the solution of the first stage if it arrives early, withhold the solution if it arrives neither too soon nor too late, and exit the competition if it has not arrived after a sufficiently long time.  From a social welfare perspective, a competition is surprisingly not always optimal: in certain scenarios, it is socially more desirable to assign the R\&D project to a single firm.

\begin{flushleft}
\end{flushleft}

\begin{flushleft}
\end{flushleft}
\end{abstract}

\newpage

\section{Introduction}

In a multi-stage R\&D process, uncertainty is a prominent feature of early but important stages. For instance, to build a revolutionary fighter jet, the manufacturing firm's research team has to first develop a new supersonic-speed engine, whose success likelihood as well as time and resource needed are often unknown. This uncertainty has to be resolved via learning over time: if the firm has spent a long time in vain on a stage, it comes to realize that this stage is probably hard, even unfeasible; on the contrary, when it solves a stage quickly it knows that the stage is solvable and most likely to be easy. When more than one firm competes for such a multi-stage R\&D project, interesting questions arise regarding the firms' strategic behavior. First, staying in the research process is costly in terms of time and resource that can be otherwise spent elsewhere for profits. In view of this, when should a firm quit the competition? Second, when a firm solves an early stage it faces the choice of whether to disclose its findings (by filing a patent, for instance). Doing so secures the relevant economic value for the firm if its competitors have not yet also completed the stage, but may also result in technological spillover so that all competitors can work on the next stage now and the firm loses its leading position in the race. On the other hand, keeping the solution private gives the firm the advantage of working on the next stage alone, but puts the payoff from the prior stage at risk because some competitor may just work it out soon enough and claim a fair share. How then, should a firm decide when to disclose a solution (or not)?

In this paper, we propose a first model to comprehensively analyze firms' strategic decisions in a dynamic R\&D competition under uncertainty, and to answer the above questions. We consider two firms competing on a continuous time line for a research project consisting of two stages, the first being an innovative stage to develop an intermediate product and the second a commercial stage to release a final product. The solution to the first stage is required for working on the second. The success on each stage arrives in a Poisson rate, which is fixed for the second stage, but is uncertain for the first stage: it is either positive (feasible) or zero (unfeasible). Whenever a firm solves a stage, it can either disclose or withhold the solution. Disclosure entails technology spillover so that both firms can start working on the next stage. Each stage has its own value, which is wholly captured by the firm which discloses the solution if its opponent has not solved it yet, and is equally shared if both firms have the solution by the time either firm discloses it. Last but not least, there is a constant research cost per unit time to work on the first stage, while a firm is free to exit the competition once and for all at any moment.

We first characterize the unique symmetric equilibrium of this game, which always has a ``disclose-withhold-exit'' pattern. There are two important time thresholds: a firm will disclose the solution to the first stage if it succeeds (in the first stage) by the first threshold, withhold the solution if it succeeds between the thresholds, and exit if it has not succeeded by the second threshold. This result stands in contrast to the general claim in existing literature (e.g. Choi\cite{Choi}, De Fraja\cite{DeFraja}, Bloch and Markowitz\cite{BM}) that the success of one firm in the intermediate product always benefits the rival as well, through either information update or technology spillover. We show that the invention of the intermediate product will never affect the rival's payoff or incentive in a late period of the competition, as long as the firms are naturally allowed to conceal their success and to exit the competition at will.

The reasoning behind our result is a firm's learning dynamics. When a firm has stayed in the competition long enough but has not received any good news from either itself or the opponent, it had better quit because information update implies that the first stage is likely to be unfeasible. In view of this, if a firm solves the first stage at a time close to the exit point, it will withhold the solution in the hope that the opponent will exit soon enough. Nevertheless, if the solution comes early there is not much benefit in withholding it and waiting, so the firm will just disclose the solution to guarantee the one bird in hand.

Secondly, we turn to the welfare perspective and find that in contrast to the conventional wisdom, a competition is not always desirable under uncertainty. In particular, we show that when the first stage is likely to be solvable and the research cost is high, the ex ante probability of completing the whole research project is higher when only one firm works on it than when two firms compete. The underlying intuition is that during a competition in such a scenario, the fear of lagging in stage 1 will ultimately outweigh the fear of stage 1 being unsolvable, causing the firms to exit prematurely compared to the case of a single firm. As a further implication on policy, if the success of a R\&D project is crucial but the success likelihood is uncertain, the relevant authorities (the Ministry of Defense, for instance) should carefully assess both the distribution of the success likelihood and the associated cost, and may consider assigning the research task to one firm privately rather than tendering it to multiple competing firms.

The remainder of this paper is organized as follows: Section 2 reviews the related literature. Section 3 presents the model. Section 4 states the main results with illustrative examples. Section 5 concludes.

\section{Literature Review}

Early economic models of R\&D races explore the investment decisions that firms make in an effort to reduce production cost\cite{DS1,DS2}, obtain a rewarding technological breakthrough\cite{Reinganum}, or secure intellectual property rights\cite{Loury,LW}. Among these seminal works, Loury\cite{Loury} proposes the framework with exponentially distributed time to innovation success, which has become one of the standard ways to depict a R\&D competition in subsequent research. We adopt this approach in our model with an added structure of uncertainty. While the early works usually emphasize a firm's optimal investment problem in a one-shot game, we focus our attention on the optimal delay or even concealment of intermediate innovation in a two-stage competition.

In the domain of R\&D races with more than one stage, a number of theoretical studies have tried to depict firms' strategic interaction at an intermediate stage. Grossman and Shapiro\cite{GS} and Harris and Vickers\cite{HV} extend Loury's model to two stages and study the competing firms' investment incentives. They find that competition is most intense when firms are close to each other in the race, in the sense that they are most willing to increase research efforts. Later studies such as Bloch and Markowitz\cite{BM} focus on disclosure delay of intermediate results like we do, but base their analysis on complete information in the research process: firms always know exactly their rate of a success in each stage following participation/investment. The general conclusion drawn from this literature is that a success in intermediate stages is beneficial for both the inventor and its opponent through technology spillover that follows disclosure. De Fraja\cite{DeFraja} comes to a similar assertion by modeling spillover as investment that will also benefit the opponent.

Our paper is closest to that of Choi\cite{Choi}, which is the first to raise the issue of success rate uncertainty. Choi's model features a two-stage race where the success rate of the first stage is unknown, and the result argues similarly to the above that one firm's success in the early stage always benefits the opponent by signaling a low difficulty level. This proposition, nevertheless, stems from the assumption that no firm can hide its success, i.e. a firm automatically knows when its opponent has solved the first stage. In contrast, our model assumes that firms are free to choose whether or not to disclose their success, which leads to an entirely different prediction that, under uncertainty, possible technology spillover only occurs when either firm solves the first stage shortly after the competition begins. After a certain time threshold, no firm will disclose its intermediate success and spillover will cease thereafter.

A most recent development in this field is Bobtcheff et al.\cite{BBM}. They also focus on two-player priority races where the solution to a valuable problem is privately observed and every player with the solution need to decide when to disclose her result. The longer a player waits, the larger the value of her solution becomes if she still preempts her opponent. As a result, players in their model behave in the opposite way to ours: they withhold the solution at the beginning, and only disclose it when it is ``mature'', i.e. the value has grown considerably after a sufficiently long time. 

There is also a large literature in law and economics that analyzes partial information disclosure in patent races, but with quite a different framework. In this line of research, a firm's does not strategically disclose information useful to its opponent to secure the immediate benefit as studied in our paper. Instead, possible incentives for information disclosure include signaling strength and commitment to the race\cite{AY,Gill}, inducing exit of risk-averse competitors\cite{BR,LBK,BM2} or establishing prior art as a defensive measure\cite{Parchomovsky}.

\section{Model}

\subsection{The R\&D Competition Game}

Consider two firms, A and B, who compete for the success of a research project on a continuous time line. To complete the project, a firm has to complete two stages of research: an innovative stage to invent a new product (such as a rocket engine for a fighter jet, or a high-performance graphics card for a gaming laptop), and a commercial stage to develop a marketable final product. The success of stage 1 is required for any research on stage 2. 

At $t=0$, they begin developing stage 1 with an i.i.d. Poisson rate of success $\lambda$. $\lambda $ may take one of two values: $H>0$ or $L=0$. In other words, a firm can either solve stage 1 with positive probability, or it can never succeed. The value of $\lambda$ cannot be observed by either firm, but has to be learned over time. For simplicity, we assume that at $t=0$, both firms hold an identical prior $\tilde{\lambda}=\alpha\in (0, 1)$. We will explain the firms' information update process in detail in the next subsection.

Once a firm solves stage 1, it can choose whether to disclose or to withhold its invention. When it discloses the solution, if its opponent has not succeeded yet, the firm claims all credit for the invention and receives a reward of $p_{1}>0$. However, if its opponent has also succeeded (but has chosen to withhold the success), the firms will ultimately receive $\frac{p_{1}}{2}$ each. In many applications of the model, this can be regarded as the scenario where the opponent files a lawsuit for the proprietorship of the invention\footnote{In principle, whether to file the lawsuit and obtain half of the reward is a choice, but here it must be the optimal action to take. Hence, without loss of generality, we assume that this process is automatic.}.

Once the solution of stage 1 is disclosed, technology spillover takes place: the product becomes available to both firms, and they then begin working on stage 2, the second and final product, with an i.i.d. Poisson rate of success $\mu$. We assume that $\mu $ only takes one value and that the value is common knowledge, thereby allowing us to focus on uncertainty in stage 1 without considering insignificant technical details. However, if a firm chooses not to disclose the solution after solving stage 1, it can still work on stage 2 on its own. Whichever firm solves stage 2 first receives a reward of $p_{2}>0$.

Staying in the competition may be costly. We assume that during stage 1, each firm pays a cost per unit time $c>0$, which can be interpreted as the opportunity cost of the resources needed to establish a new department, hire external experts, etc. On the other hand, the cost for the commercial stage is much lower than innovation, because it is part of the firm's regular business and does not require the reallocation of resources. Hence, we assume for convenience that the cost is zero for stage 2. Once stage 2 is solved, the game ends. A firm can choose to exit the competition at any time point before it ends; once a firm exits, it cannot re-enter the competition. Moreover, we assume that exit cannot be observed by the other firm.

Finally, we make the following three assumptions on the set of parameters $\{\alpha, H, \mu, p_1, p_2, c\}$ throughout the paper. We will be more specific about the use of each assumption in the analysis section.

\textbf{A1.} $\alpha H (p_1+p_2)>c$. This assumption guarantees that at least at $t=0$ each firm has a positive rate of net return.

\textbf{A2.} $p_1H>p_2\mu$. This assumption makes it possible for disclosure to occur in equilibrium.

\textbf{A3.} The following inequality is satisfied:
\begin{align*}
&\frac{H(p_{1}+p_{2})[\frac{\frac{1}{2}H-\mu }{H-\mu }(\frac{H(p_1+p_2)}{p_1H-p_2\mu})^{-\frac{H}{H+\mu}}+\frac{\frac{1}{2}H}{H-\mu }(\frac{H(p_1+p_2)}{p_1H-p_2\mu})^{-\frac{\mu}{H+\mu}}]}{c}+\frac{\mu }{H-\mu }(\frac{H(p_1+p_2)}{p_1H-p_2\mu})^{-\frac{H}{H+\mu}}\\
&-\frac{H}{H-\mu }(\frac{H(p_1+p_2)}{p_1H-p_2\mu})^{-\frac{\mu}{H+\mu}}\geq (\frac{H(p_1+p_2)}{p_1H-p_2\mu})^{\frac{H}{H+\mu}}.
\end{align*}
This assumption allows for an equilibrium where disclosure is weakly optimal for each firm at $t=0$ when $\alpha=\frac{1}{2}$, i.e. a firm has equal prior beliefs on the possible values of $\lambda$.

\subsection{Strategy, Information and Equilibrium}

Before solving stage 1, it is clear that a firm is only free to choose when to exit. After finishing the intermediate product, by our assumption above a firm never exits. It then conditions its decision of disclosing/withholding the solution on time elapsed. We assume that a firm uses Bayesian updating whenever possible.

In this paper we focus on symmetric strategies that are piecewise continuous. That is, a \textbf{strategy profile} can be characterized by a number of cutoffs $t_{1},t_{2}...t_{n}$ in the following way: for an arbitrary $i$, each firm's action at time $t$ is the same for every $t\in (t_{i},t_{i+1})$. We call $(t_{i},t_{i+1})$ a "disclose region" if
the firm immediately discloses any incoming success in $(t_{i},t_{i+1})$; we call $(t_{i},t_{i+1})$ a "withhold region" if the firm withholds any
incoming success in $(t_{i},t_{i+1})$; and we call $(t_{i},t_{i+1})$ an exit region if the firm exits.

The key difference made by this model is the existence of both uncertainty and the firms' strategic response to it. Firms do not directly observe $\lambda$, and is free to conceal their success in stage 1 to partly manipulate the opponent's belief. Formally, each firm holds a belief $\tilde{\lambda}$ regarding the difficulty of the project $\tilde{\lambda}=\Pr (\lambda =H)$\footnote{Since we focus on the symmetric case, for notational convenience we choose not to use different labels here for the two firms' beliefs}. A firm updates its belief through (1) its own progress (success or not) in stage 1 and (2) its opponent's disclosed success or silence. It is important to note here that the belief updating process of each firm is affected by its opponent's strategy.

To demonstrate how a firm's belief about $\lambda$ evolves over time, we consider the following simple symmetric case: each firm discloses its success whenever it solve stage 1 by time $t_1$, but withhold its success from $t_1$ to another time point $t_2$, after which it exits the competition if it has not solved stage 1. As we will show in the result section, this is exactly how the firms will behave in the unique equilibrium. 

We write $\tilde{\lambda}$ as $\tilde{\lambda}(t)$, a function of time $t$. The trajectory of $\tilde{\lambda}(t)$ is simple in the disclose region $(0,t')$. Absent a disclosure of success, $\tilde{\lambda}(t)$ evolves as follows:
\begin{equation}
\tilde{\lambda}(t)=\frac{\alpha e^{-2Ht}}{\alpha e^{-2Ht}+1-\alpha} \label{2}.
\end{equation}
On the other hand, whenever either firm succeeds in this region, both firms' beliefs jump to 1 because the successful firm will disclose the solution immediately.

In the withhold region, the updating process is somewhat more complex, as silence does not explicitly suggest whether the opponent has solved stage 1 or not; instead, it implies that the opponent has not solved stage 2, thus implicitly undermining the likelihood of its success in stage 1. When $t\in (t_1,t_2)$, absent a disclosure of success, $\tilde{\lambda}(t)$ evolves as follows:
\begin{equation}
\tilde{\lambda}(t)=\frac{\alpha e^{-2Ht_{1}-H(t-t_{1})}(e^{-H(t-t_{1})}+\int_{0}^{t-t_1}He^{-sH}e^{-(t-t_1-s)\mu }ds)}{%
\alpha e^{-2Ht_{1}-H(t-t_{1})}(e^{-H(t-t_{1})}+\int_{0}^{t-t_1}He^{-sH}e^{-(t-t_1-s)\mu }ds)+1-\alpha}. \label{4}
\end{equation}
Whenever a firm succeeds in this region, only its own belief jumps to 1 since it will keep the good news private.

Intuitively, $\tilde{\lambda}(t)$ should be decreasing as waiting in vain for a success can only indicate that stage 1 is more and more likely to be unfeasible. This is confirmed by the following lemma.

\begin{lem}
$\tilde{\lambda}(t)$ always decreases in $t$.
\end{lem}

\begin{proof}
For (\ref{2}), observe that $e^{-2Ht}$ decreases in $t$; thus, $\tilde{\lambda}(t)$ decreases in $t$.

For (\ref{4}), when $\mu\neq H$ we have
\begin{align*}
e^{-H(t-t_{1})}+\int_{0}^{t-t_1}He^{-sH}e^{-(t-t_1-s)\mu }ds&=e^{-H(t-t_1)}+\frac{H}{\mu-H}(e^{-(t-t_1)H}-e^{-(t-t_1)\mu})\\
&=\frac{\mu}{\mu-H}e^{-H(t-t_1)}-\frac{H}{\mu-H}e^{-(t-t_1)\mu}.
\end{align*}
The derivative of the right-hand side with respect to $t$ is equal to
\begin{align*}
-\frac{\mu H}{\mu-H}(e^{-H(t-t_1)}-e^{-\mu(t-t_1)}),
\end{align*}
which is always negative. Hence, we can conclude that $e^{-2Ht_{1}-H(t-t_{1})}(e^{-H(t-t_{1})}+\int_{0}^{t-t_1}He^{-sH}e^{-(t-t_1-s)\mu }ds)$ decreases in $t$ and thus $\tilde{\lambda}(t)$ decreases in $t$.

When $\mu=H$, we have
\begin{align*}
e^{-H(t-t_{1})}+\int_{0}^{t-t_1}He^{-sH}e^{-(t-t_1-s)\mu }ds&=e^{-H(t-t_{1})}+H(t-t_1)e^{-H(t-t_{1})},
\end{align*}
which again decreases in $t$. Therefore, $\tilde{\lambda}(t)$ decreases in $t$.
\end{proof}

It is simple yet useful to understand the reasoning behind Lemma 1. In the disclose region, absent any disclosure, a firm's belief update is twofold. First, its unsuccessful research decreases its estimation of $\lambda $; second, the implied public information that its opponent has also not succeeded again decreases its estimation of $\lambda $. In the withhold region, although the firm no longer observes its opponent's progress, it knows at least that the opponent has not completed stage 2 (otherwise, the opponent will disclose everything and win the game) and thus is also less likely to have completed stage 1. All this information diminishes the likelihood that stage 1 can ever be solved, i.e., $\tilde{\lambda}(t)$ always decreases in $t$.

At the end of this section, we define the equilibrium of this game. Following the conventional characterization, an \textbf{equilibrium} is a (symmetric, piecewise continuous) strategy profile such that for each firm, at any time $t\geq 0$, following the action specified in the strategy profile maximizes the firm's expected continuation payoff given the strategy of its opponent.

\section{Results}

\subsection{Characterization of Unique Equilibrium}

Our first main result is the explicit characterization of the unique equilibrium. Noting the nature of sequential optimality in an equilibrium, our approach is to work backwards, i.e. we determine a firm's optimal action in the following order: (1) when it has not solved stage 1 after spending a long time in the competition, (2) when it solves stage 1 only at a late time, and (3) when it solves stage 1 early. We begin with the following lemma.

\begin{lem}
In every equilibrium, there exists $\bar{t}$ such that a firm exits when $t\in \lbrack \bar{t},\infty ]$.
\end{lem}

\begin{proof}
As $t\rightarrow \infty $, absent success, $\tilde{\lambda}(t)<\frac{\alpha e^{-Ht}}{\alpha e^{-Ht}+1-\alpha}\rightarrow 0$. Hence as $\tilde{\lambda}(t)>0$, we have $\tilde{\lambda}(t)\rightarrow 0$. Then, we compare the firm's incentives for staying and exiting. The incentive to stay an additional time increment $dt$ is less than $\tilde{\lambda}(t)H(p_{1}+p_{2})dt$, which approaches $o(dt)$ as $t\rightarrow \infty $. By our assumption of a fixed $c>0$, an exit point always exists.
\end{proof}

The intuition for the existence of an exit region in every equilibrium is that no firm is willing to stay indefinitely in the competition when staying is costly and stage 1 may just be unsolvable. As time goes by with no good news arriving, a firm believes that $\lambda=0$ is more and more likely and will ultimately decides to exit. Recall that every strategy profile, and hence every candidate for an equilibrium, can be considered here can be characterized by a number of cutoffs $t_{1},t_{2}...t_{n}$. Therefore $(t_{n},\infty )$ must be the only exit region in every equilibrium.

Next, we go one step backwards to show that withholding the solution must be optimal when time gets close to the exit point.

\begin{lem}
In every equilibrium, each firm must be withholding the solution in $(t_{n-1},t_n)$.
\end{lem}

\begin{proof}
Suppose, alternatively, that in some equilibrium both firms choose to disclose the solution at $t_{n}-dt$. Consider either of the firms, say firm A, and suppose that the probability that the firm's opponent has succeeded in stage 1 before this time instant is $p$. The expected payoff from disclosure is $p_{1}(1-p)+\frac{1}{2}p_{1}p+\frac{1}{2}p_{2}$, while the expected payoff from withholding the solution is $(p_{1}+p_{2})(1-p-o_1(t))+\frac{1}{2}(p_{1}+p_{2})(p+o_2(t))$. In the second expression, $o_1(t)$ and $o_2(t)$ represent firm A's and firm B's (infinitesimal) probabilities of solving stage 1 during $(t_n-dt,t_n)$ respectively, calculated from firm A's belief, with $o_1(t)=o_2(t)$ due to the assumed symmetry of strategies. It is clear that the value of the second expression is larger than that of the first.
\end{proof}

When a firm solves stage 1 just before the exit point, it compares its expected payoff from two available actions. If the opponent has solved stage 1 by now, disclosing or withholding the solution does not make any difference. However, if the opponent has been unsuccessful so far, disclosure would have ``saved'' the opponent from the edge of exiting and brought the competition to stage 2, while withholding the solution will probably force the opponent to quit because it is unlikely to solve stage 1 in the last brief period. Therefore, withholding the solution is no doubt the wise choice. Furthermore, we show below that no firm will choose to withhold the solution in two different time regions.

\begin{lem}
In any equilibrium, there exists one and only one withhold region.
\end{lem}

\begin{proof}
By Lemma 3, there must be at least one withhold region. Suppose that there is a series of withhold regions, $(t_{1},t_{2})$, $(t_{3},t_{4})$, ...,$(t_{2m-1},t_{2m})$, $m\in\mathbb{N}^+$. Without loss of generality, we analyze the firm's incentives at $t_{1}$. 

The payoff from disclosure is $p_{1}+\frac{p_{2}}{2}$, which is a constant. The expected payoff from withholding the solution from $t_{1}$
to $t_{2}$ is
\begin{eqnarray*}
&&e^{-H(t_{2}-t_{1})}e^{-\mu (t_{2}-t_{1})}(p_{1}+\frac{p_{2}}{2}%
)+\int_{t_{1}}^{t_{2}}\mu e^{-\mu (s-t_{1})}e^{-H(s-t_{1})}ds
(p_{1}+p_{2}) \nonumber \\
&&+\int_{t_{1}}^{t_{2}}He^{-H(s-t_{1})}e^{-\mu (s-t_{1})}ds 
\frac{(p_{1}+p_{2})}{2} \nonumber \\
&=&[-\frac{\mu }{H+\mu }e^{-H(s-t_{1})-\mu
(s-t_{1})}|_{t_{1}}^{t_{2}}](p_{1}+p_{2})+e^{-H(t_{2}-t_{1})-\mu
(t_{2}-t_{1})}(p_{1}+\frac{p_{2}}{2}) \nonumber \\
&&+[-\frac{H}{H+\mu }e^{-H(s-t_{1})-\mu (s-t_{1})}|_{t_{1}}^{t_{2}}]\frac{%
(p_{1}+p_{2})}{2} \nonumber \\
&=&p_{1}[\frac{\frac{1}{2}H}{H+\mu }e^{-(H+\mu )(t_{2}-t_{1})}+\frac{\frac{1%
}{2}H+\mu }{H+\mu }]+p_{2}[\frac{-\frac{1}{2}\mu }{H+\mu }e^{-(H+\mu
)(t_{2}-t_{1})}+\frac{\frac{1}{2}H+\mu }{H+\mu }].
\end{eqnarray*}

At cutoff $t_{1}$, the firm should be at least indifferent between disclosing and withholding the solution:
\begin{eqnarray*}
p_{1}\frac{\frac{1}{2}H}{H+\mu }[e^{-(H+\mu )(t_{2}-t_{1})}-1]+p_{2}\frac{%
\frac{1}{2}\mu }{H+\mu }[1-e^{-(H+\mu )(t_{2}-t_{1})}] &\geq&0 \\
p_{1}\frac{\frac{1}{2}H}{H+\mu } &\leq&p_{2}\frac{\frac{1}{2}\mu }{H+\mu },
\end{eqnarray*}%
which contradicts our assumption $p_{1}H>p_{2}\mu $.
\end{proof}

Our approach for the proof here is to assume multiple withhold regions in an equilibrium and evaluate the incentives for different actions at the instant of entering the first such region. We show that the condition needed for starting the withhold region now contradicts the very assumption that guarantees the existence of a disclose region. Therefore, the withhold region in every equilibrium is unique and withholding the solution must be followed by exit.

Now we are ready to present our main theorem.

\begin{thm}
The R\&D competition game has a unique equilibrium. The equilibrium is a disclose-withhold-exit strategy profile characterized by cutoffs $t_{1}$ and $t_{2}$ such that each firm
\begin{eqnarray*}
\text{discloses the solution of stage 1 if }t &\in &[0,t_{1}] \nonumber \\
\text{withholds the solution of stage 1 if }t &\in &[t_{1},t_{2}] \nonumber \\
\text{exits if }t &\in &[t_{2},\infty ].
\end{eqnarray*}
\end{thm}

\begin{proof}
See Appendix A.
\end{proof}

This result highlights our distinction with previous literature, that the positive externalities from technology spillover may never realize because firms change their behavior according to time and belief. Whenever a firm solves stage 1, it immediately realizes that stage 1 is solvable, but its decision of whether to disclose the solution depends on when this jump of belief takes place. If it succeeds early in the competition, there is still a long time before the opponent quits; given that stage 1 is solvable, the opponent may just also solve stage 1 before the exit point. Hence the rational choice to make is to disclose the solution at once so that $p_1$ is secured. Here technology spillover will occur and the overall progress of R\&D is indeed accelerated (as compared to, for instance, the case where the two firms work independently on the project without any technology spillover).

However, when no firm has solved stage 1 before the time threshold $t_1$, the incentives for disclosing and withholding the solution change drastically. Now that not much time remains before the exit point, it becomes unlikely that the opponent can succeed in stage 1 if it has not already. Hence, the benefit of disclosure has not increased -- in fact, it has decreased since the opponent may have solved stage 1 some time between $t_1$ and now -- but there is a larger cost now because disclosure could save the opponent from quitting and simply cut the expected payoff from the next stage by half. As a result, both firms would rather conceal any success and wait for the opponent to quit. Technology spillover will never happen in this time region and a competition is no different from pursuing R\&D independently.

\subsection{Welfare Analysis}

In this section we analyze the effect of uncertainty on social welfare. In particular, we are interested in whether the conventional wisdom that competition always facilitates efficiency still holds here: under uncertainty, is a two-firm competition always socially preferred to a one-firm monopoly? To provide a tractable answer, we assume that the final outcome is indispensable to the society, meaning that, from a social welfare perspective, a success is always desirable, while the associated cost is relatively trivial. Thus, social welfare can be measured by the probability that the game ends in solving stage 1 (by at least one firm) rather than in all firms exiting. This notion is consistent with the social welfare analysis in a wide range of literature in economics and law which studies innovation and patent races (see for example \cite{BR}\cite{Parchomovsky}).

The above probability has a one-to-one relationship with the total length of
time elapsed per firm before all firms exit. In the case of multiple firms,
this time length is the sum of individual time lengths. Hence, our approach
here is to calculate it in an R\&D process without competition (a single
firm), and one with competition (two firms) respectively.

When there is only one firm, information is perfect: the firm always knows
whether or not it has solved stage 1. With two competing firms, however,
information is imperfect in the sense that no firm knows whether its
opponent has solved stage 1 once they enter the withhold region. In general,
the imperfect information has two effects:

\textbf{Effect 1:} A firm's expected payoff from solving stage 1 decreases. For a
single firm, a success in stage 1 brings the reward in whole; however, for two competing firms
in the withhold region, a firm is faced with the fear that the competitor has already
solved the first stage, meaning that a success in stage 1 will only yield $\frac{p_{1}+p_{2}%
}{2}$. From this aspect, a firm is less willing to stay in the competition.

\textbf{Effect 2:} A firm's belief about stage 1 being solvable decreases
slower. For a single firm, all the unfruitful effort reduces it's belief $%
\tilde{\lambda}(t)$. However, for two competing firms in the withhold region, a
firm is uncertain about the result of the competitor's effort., making it
less pessimistic about the difficulty of stage 1. From this aspect, a firm
is more willing to stay in the competition.

When the first effect is larger than (equal to, smaller than) the second,
competition shortens (does not change, extends) the total time spent on R\&D without solving stage 1 and hence decreases (does not change, increases) social welfare. The
formal statement is as follows.

\begin{prop}
From the welfare perspective,%
\[
\left\{ 
\begin{array}{c}
\text{the society prefers a one-firm monopoly if }\frac{H(p_{1}+p_{2})}{2}%
<c; \\ 
\text{the society is indifferent if }\frac{H(p_{1}+p_{2})}{2}=c; \\ 
\text{the society prefers a two-firm competition if }\frac{H(p_{1}+p_{2})}{2}>c.%
\end{array}%
\right. 
\]%
\end{prop}

\begin{proof}
See Appendix B.
\end{proof}

Besides comparing the total time on R\&D in the two schemes as in our formal proof, we can also derive the result from the following condition:
\begin{align*}
c=H\cdot \frac{(p_{1}+p_{2})e^{-H(t_{2}-t_{1})}+\frac{p_{1}+p_{2}}{2}%
\int_{0}^{t_{2}-t_{1}}He^{-Hs}e^{-(t_{2}-t_{1}-s)\mu }ds}{%
e^{-H(t_{2}-t_{1})}+\int_{0}^{t_{2}-t_{1}}He^{-Hs}e^{-(t_{2}-t_{1}-s)\mu }ds+%
\frac{1-\alpha }{\alpha }e^{Ht_{1}+Ht_{2}}}.
\end{align*}
This is the exit condition of a firm in a competition which can withhold its solution to stage 1. On the other hand, a hypothetical exit condition for a firm, if information became perfect after $t_1$, is (we use $t_2'$ here to distinguish the hypothetical exit time from the original $t_2$)
\begin{align*}
c=H\cdot \frac{(p_{1}+p_{2})e^{-H(t'_{2}-t_{1}))}}{e^{-H(t'_{2}-t_{1})}+\frac{1-\alpha }{\alpha }e^{Ht_{1}+Ht'_{2}}}.
\end{align*}
The second scenario produces the same total length of time on R\&D without solving stage 1 as the case with a one-firm monopoly. In the first scenario (our original R\&D competition game), a new state, where the
opponent has finished stage 1 but not stage 2, is incorporated into the
update, introducing $\int_{0}^{t_{2}-t_{1}}He^{-Hs}e^{-(t_{2}-t_{1}-s)\mu }ds
$ on the denominator and $\frac{H(p_{1}+p_{2})}{2}%
\int_{0}^{t_{2}-t_{1}}He^{-Hs}e^{-(t_{2}-t_{1}-s)\mu }ds$ on the numerator.
The new term on the numerator is $\frac{H(p_{1}+p_{2})}{2}$ times the new
term on the denominator; thus which of $t_2$ and $t_2'$ is larger depends on
the value of $\frac{H(p_{1}+p_{2})}{2}$ and the original value of the
fraction, $c$.

The intuition behind this result is as follows. When $c$ is high, a firm in a competition
exits sooner, i.e., it's belief $\tilde{\lambda}(t)$ is still
relatively high when exiting. Thus, by the end of the withhold region, as
time passes by, the firm is not yet so concerned about the possibility that
the task is impossible (Effect 2); rather, it is scared away as its opponent
is increasingly likely to have finished the first stage (Effect 1). As
Effect 1 plays the decisive role, two firms in a competition are less persistent
than in a monopoly.

On the other hand, when $c$ is low, a firm exits later when it's belief $\tilde{\lambda}(t)$ is relatively low. Thus, by the end of the withhold region, the
firm is less worried that its opponent may have solved stage 1, but is more concerned that stage 1 may not be solvable after all. In contrast to the above situation, Effect 1 is now very small and Effect 2 dominates. As we have discussed, $\tilde{\lambda}(t)$
evolves more slowly due to imperfect information, and hence two firms in a competition are more persistent than in a monopoly.
\vspace{3ex}

\noindent\textbf{Numerical Example}

We assume that the arrival rate of a success is $H=1$, $L=0$ for the first
stage and $\mu =0.5$ for the second stage, with prior $\alpha=0.8$, cost of research $c=0.8$, revenue from the intermediate good $p_{1}=1$%
, and revenue from the final good $p_{2}=0.2$. The numerical
result is depicted by Figures 1 and 2.

\begin{figure}[h]
\centering
\includegraphics[width=5.5in]{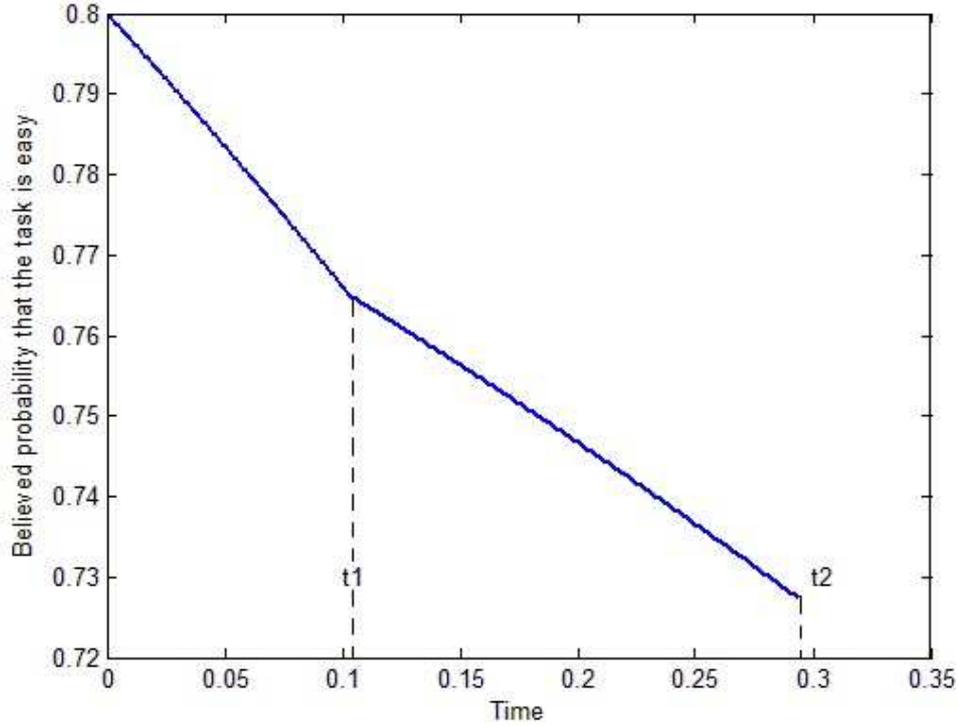}
\caption{Trajectory of belief}
\end{figure}

\begin{figure}[h]
\centering
\includegraphics[width=5.5in]{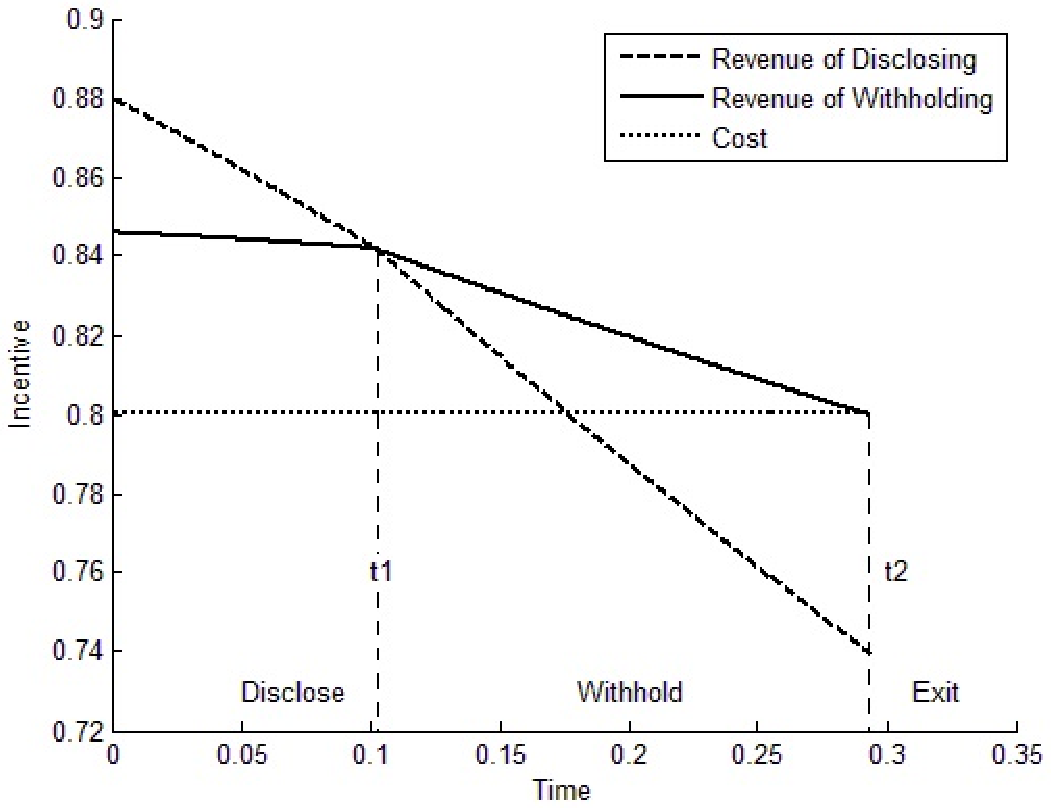}
\caption{Trajectory of cost and expected benefit}
\end{figure}

Figure 1 shows the trajectory of the firm's belief regarding the difficulty
of the first stage, i.e. $\tilde{\lambda}(t)$; Figure 2 illustrates the comparison
of a firm's revenue from "disclosing the solution now" and from "withholding the
solution until the end", provided that the firm has invented the
intermediate product. The firm exhibits disclose-withhold-exit behavior with
cutoffs $t_{1}=0.103$ and $t_{2}=0.294$, i.e., the disclose region is $%
[0,0.103]$, the withhold region is $[0.103,0.294]$, and the exit point is $%
0.294$. Thus, the total time spent on research by the two firms without solving stage 1 is $0.588$. For a
monopolist, the total time of research without is $0.693$. From the
society's perspective, a monopoly is preferred.

For comparison purpose, if we lower $c$ to $0.2$, i.e. smaller than $\frac{%
H(p_{1}+p_{2})}{2}$, the cutoffs become $t_{1}=1.346$ and $t_{2}=1.537$; in the two-firm competition, the
total time spent on research without solving stage 1 is $3.074$. However in such a case
a monopolist shall exit at $t=2.996$. From the society's perspective,
a competition is preferred.
\vspace{3ex}

\noindent\textbf{Discussion}

In scenarios where the final output is not indispensable, the implication of our model
varies. Suppose that the society values the final output at $%
p$, $p>p_{1}+p_{2}$. The socially optimal exit point is then%
\begin{align*}
\hat{t}=\max\{-\frac{\ln (\frac{c}{Hp-c}\frac{1-\alpha }{\alpha })}{H}, 0\}.
\end{align*}
for a monopoly firm, or $\frac{\hat{t}}{2}$ for a firm in a competition. Our model implies the following result on the socially desirable scheme:

If $\frac{H(p_{1}+p_{2})}{2}>c$, the society always prefers a monopoly.

If $\frac{H(p_{1}+p_{2})}{2}=c$, the society is indifferent between
the two.

If $\frac{H(p_{1}+p_{2})}{2}<c$, the total time spent on research in a two-firm competition is longer; in this
case, which scheme the society prefers depends on the expected value of the final output after a monopoly's exit point, as well as the difference between the total time spent on research under the two schemes.

\section{Conclusion}

Uncertainty in the difficulty of research is a common and important feature of an R\&D process. In this paper, we have proposed a model to analyze how uncertainty affects a firm's strategic behavior in a two-stage R\&D competition. The competing firms are free to choose whether and when to disclose their solution of the intermediate stage, and also whether and when to exit the entire competition. We find that this game has a unique symmetric equilibrium featuring two cutoffs of time: each firm will disclose its solution if it solves the intermediate stage by the first cutoff, withhold its solution between the two cutoffs, and exit if it has not solved the intermediate stage by the second cutoff. When the R\&D project is highly valuable to the society and welfare can be measured by the probability of completing the whole project, we show that a competition is not always the desired scheme:  the society may benefit from assigning the project to one single firm.

We believe that this paper may open the way for richer studies on R\&D competitions with uncertainty and related policy issues. One natural extension of the model is to solve the game with a larger number of firms and/or firms of different research capacity and cost. Also, there may be more than one research path to reach the final stage, and it is interesting to explore how competing firms choose to experiment and learn about distinct paths.

\newpage
\bibliographystyle{acm}
\bibliography{reference}

\newpage
\appendix
\section{Proof of Theorem 1}

Combining Lemmas 2, 3, and 4, we know that every equilibrium is a
disclose-withhold-exit equilibrium. Now we characterize the unique cutoffs $t_1, t_2$ and prove that the corresponding strategy profile is an equilibrium. The proof consists of four parts.

\textbf{Part I: characterize $t_1$ and $t_2$.} In the disclose region, the payoff from disclosure is $p_1+\frac{p_{2}}{2}$; on the other hand, the expected payoff from withholding the solution from $t_1$ to $t_2$ is
\begin{eqnarray}
&&e^{-H(t_{2}-t_{1})}e^{-\mu
(t_{2}-t_{1})}(p_{1}+p_{2})+\int_{t_{1}}^{t_{2}}\mu e^{-\mu
(s-t_{1})}e^{-H(s-t_{1})}ds (p_{1}+p_{2}) \nonumber \\
&&+\int_{t_{1}}^{t_{2}}He^{-H(s-t_{1})}e^{-\mu (s-t_{1})}ds 
\frac{(p_{1}+p_{2})}{2} \nonumber \\
&=&[-\frac{\mu }{H+\mu }e^{-H(s-t_{1})-\mu
(s-t_{1})}|_{t_{1}}^{t_{2}}+e^{-H(t_{2}-t_{1})-\mu
(t_{2}-t_{1})}](p_{1}+p_{2}) \nonumber \\
&&+[-\frac{H}{H+\mu }e^{-H(s-t_{1})-\mu (s-t_{1})}|_{t_{1}}^{t_{2}}]\frac{%
(p_{1}+p_{2})}{2} \nonumber \\
&=&\frac{(p_{1}+p_{2})}{2}[\frac{H}{H+\mu }e^{-(H+\mu )(t_{2}-t_{1})}+\frac{%
H+2\mu }{H+\mu }]. \nonumber
\end{eqnarray}

At cutoff $t_{1}$, the firm should be indifferent between disclosing and
withholding the solution:
\begin{eqnarray*}
\frac{(p_{1}+p_{2})}{2}[\frac{H}{H+\mu }e^{-(H+\mu )(t_{2}-t_{1})}+\frac{%
H+2\mu }{H+\mu }] &=&p_{1}+\frac{p_{2}}{2} \\
t_{2}-t_{1} &=&\frac{\ln \frac{\frac{p_{1}+\frac{p_{2}}{2}}{\frac{p_{1}+p_{2}%
}{2}}-\frac{H+2\mu }{H+\mu }}{\frac{H}{H+\mu }}}{-(H+\mu )} \\
t_{2}-t_{1} &=&\frac{\ln \frac{H(p_{1}+p_{2})}{p_{1}H-p_{2}\mu }}{H+\mu}.
\end{eqnarray*}%
The above equation implies that the difference between $t_{1}$ and $t_{2}$, or the time length of the withhold region, is constant when $t_2$ is sufficiently large. Thus, given any $t_{2}$, we can always find a unique $t_1$ (when $t_2<\frac{\ln \frac{H(p_{1}+p_{2})}{p_{1}H-p_{2}\mu }}{H+\mu}$, simply let $t_1=0$).

The next step is to show that $t_2$ is unique as well. Given that $\lambda=H$, through time length $t-t_{1}$, the probability that the opponent has not solved stage 1 is $e^{-H(t-t_{1})}$, the probability that it has solved stage 1 but has not solved stage 2 is $\int_{0}^{t-t_{1}}He^{-Hs}e^{-(t-t_{1}-s)\mu }ds$. Thus at time $t>t_{1}$, if the game hasn't been terminated by the opponent's success in stage 2, by Bayesian update, the exit condition is:%
\begin{equation}
c=\tilde{\lambda}(t)H\frac{(p_{1}+p_{2})e^{-H(t-t_{1})}+\frac{(p_{1}+%
p_{2})}{2}\int_{0}^{t-t_{1}}He^{-Hs}e^{-(t-t_{1}-s)\mu }ds}{%
e^{-H(t-t_{1})}+\int_{0}^{t-t_{1}}He^{-Hs}e^{-(t-t_{1}-s)\mu }ds}.  \label{22}
\end{equation}

Next we show that the RHS of (\ref{22}) is decreasing in $t$:%
\begin{eqnarray}
&&\tilde{\lambda}(t)H\frac{(p_{1}+p_{2})e^{-H(t-t_{1})}+\frac{(p_{1}+%
p_{2})}{2}\int_{0}^{t-t_{1}}He^{-Hs}e^{-(t-t_{1}-s)\mu }ds}{%
e^{-H(t-t_{1})}+\int_{0}^{t-t_{1}}He^{-Hs}e^{-(t-t_{1}-s)\mu }ds}  \nonumber
\\
&=&\tilde{\lambda}(t)H\{\frac{(p_{1}+p_{2})}{2}+\frac{\frac{(p_{1}+%
p_{2})}{2}e^{-H(t-t_{1})}}{e^{-H(t-t_{1})}+\frac{H}{\mu -H}%
[e^{-H(t-t_{1})}-e^{-(t-t_{1})\mu }]}\}  \nonumber \\
&=&\tilde{\lambda}(t)H\{\frac{(p_{1}+p_{2})}{2}+\frac{\frac{(p_{1}+%
p_{2})}{2}}{1+\frac{H}{\mu -H}[1-e^{(H-\mu )(t-t_{1})}]}\}.  \label{26}
\end{eqnarray}%
If $H>\mu $, then $1-e^{(H-\mu )(t-t_{1})}<0$ and $\mu -H<0$. As $t$
increases, $1-e^{(H-\mu )(t-t_{1})}$ decreases; $\frac{H}{\mu -H}%
[1-e^{(H-\mu )(t-t_{1})}]$ increases; (\ref{26}) decreases. If $H<\mu $, then $%
1-e^{(H-\mu )(t-t_{1})}>0$ and $\mu -H>0$. As $t$ increases, $1-e^{(H-\mu
)(t-t_{1})}$ increases; $\frac{H}{\mu -H}[1-e^{(H-\mu )(t-t_{1})}]$
increases; (\ref{26}) decreases.

As the RHS of (\ref{22}) is decreasing in $t$, also by Lemma 3, (\ref{22}) has a unique solution. We propose that this solution is $t_{2}$ by proving that the firm won't deviate (i.e. won't stay in the game) after $t_{2}$. We evaluate the firm's incentive at time $t_{2}+\Delta t$. Staying in the game yields an instantaneous rate of expect payoff:%
\begin{eqnarray}
&\tilde{\lambda}(t)H\frac{(p_{1}+p_{2})e^{-H(t_{2}-t_{1})}+\frac{%
(p_{1}+p_{2})}{2}\int_{0}^{t_{2}-t_{1}}He^{-Hs}e^{-(t_{2}+\Delta
t-t_{1}-s)\mu }ds}{e^{-H(t_{2}-t_{1})}+%
\int_{0}^{t_{2}-t_{1}}He^{-Hs}e^{-(t_{2}+\Delta t-t_{1}-s)\mu }ds}  \nonumber
\\
=&\alpha H\frac{e^{-2Ht_{1}-H(t_{2}+\Delta t-t_{1})}[(p_{1}+p%
_{2})e^{-H(t_{2}-t_{1})}+\frac{(p_{1}+p_{2})}{2}%
\int_{0}^{t_{2}-t_{1}}He^{-Hs}e^{-(t_{2}+\Delta t-t_{1}-s)\mu }ds]}{%
\alpha e^{-2Ht_{1}-H(t_{2}+\Delta
t-t_{1})}[e^{-H(t_{2}-t_{1})}+\int_{0}^{t_{2}-t_{1}}He^{-Hs}e^{-(t_{2}+%
\Delta t-t_{1}-s)\mu }ds]+1-\alpha}  \nonumber \\
=&H\frac{(p_{1}+p_{2})}{2} \nonumber \\
&\{1+\frac{e^{-2Ht_{1}-H(t_{2}+\Delta
t-t_{1})}e^{-H(t_{2}-t_{1})}-\frac{1-\alpha}{\alpha}}{e^{-2Ht_{1}-H(t_{2}+\Delta
t-t_{1})}\{e^{-H(t_{2}-t_{1})}+\frac{H}{\mu -H}[e^{-H(t_{2}-t_{1})}-e^{-\mu
(t_{2}-t_{1})}]e^{-\Delta t\mu }\}+\frac{1-\alpha}{\alpha}}\}.  \label{27}
\end{eqnarray}

At $\Delta t=0$, (\ref{26})=(\ref{27}), thus the instantaneous payoff rate is
continuous at $t_{2}$. As $\Delta t$ increases, $e^{-2Ht_{1}-H(t_{2}+\Delta
t-t_{1})}e^{-H(t_{2}-t_{1})}-1$ decreases and is negative, the term
\begin{align*}
e^{-2Ht_{1}-H(t_{2}+\Delta t-t_{1})}\{e^{-H(t_{2}-t_{1})}+\frac{H}{\mu -H}[e^{-H(t_{2}-t_{1})}-e^{-\mu (t_{2}-t_{1})}]e^{-\Delta t\mu }\}
\end{align*}
decreases and is positive, thus (\ref{27}) decreases in $\Delta t$. As the exit
condition holds at $\Delta t=0$ (i.e. $t=t_{2}$), the instantaneous payoff rate
is smaller than $c$ for all $\Delta t>0$.

\textbf{Part II: prove that a firm will never exit before $t_2$.} The above argument has proved that exit is optimal after $t_2$. For $t\in \lbrack 0,t_{1}]$, the instantaneous payoff rate from staying for another $dt$ and disclosing the solution if stage 1 is solved is
\begin{eqnarray}
\tilde{\lambda}(t)[H(p_1+\frac{p_{2}}{2})+H\frac{p_{2}}{2}]=\tilde{\lambda}(t)H(p_{1}+p_{2}).
\label{19}
\end{eqnarray}
By (\ref{22}) we know that (\ref{19}) must be greater than $c$.

For $t\in (t_{1}, t_2]$, the instantaneous payoff rate from continuing research (and adopting a
withhold strategy) before solving stage 1, rather than exiting immediately, is%
\begin{eqnarray}
&&\frac{%
\begin{array}{c}
\tilde{\lambda}(t)H\{\frac{p_{1}+p_{2}}{2}%
\int_{t_{1}}^{t}He^{-H(s-t_{1})}e^{-\mu
(t-s)}ds+e^{-H(t-t_{1})}[(p_{1}+p_{2})e^{-H(t_{2}-t)}e^{-\mu (t_{2}-t)} \\ 
+(p_{1}+p_{2})\int_{t}^{t_{2}}\mu e^{-\mu (s-t)}e^{-H(s-t)}ds+\frac{%
p_{1}+p_{2}}{2}\int_{t}^{t_{2}}He^{-H(s-t)}e^{-\mu (s-t)}ds]\}%
\end{array}%
}{\int_{t_{1}}^{t}He^{-H(s-t_{1})}e^{-\mu (t-s)}ds+e^{-H(t-t_{1})}} 
\nonumber \\
&=&\frac{%
\begin{array}{c}
H\{\frac{p_{1}+p_{2}}{2}\int_{t_{1}}^{t}He^{-H(s-t_{1})}e^{-\mu
(t-s)}ds+e^{-H(t-t_{1})}[(p_{1}+p_{2})e^{-H(t_{2}-t)}e^{-\mu (t_{2}-t)} \\ 
+(p_{1}+p_{2})\int_{t}^{t_{2}}\mu e^{-\mu (s-t)}e^{-H(s-t)}ds+\frac{%
p_{1}+p_{2}}{2}\int_{t}^{t_{2}}He^{-H(s-t)}e^{-\mu (s-t)}ds]\}%
\end{array}%
}{\int_{t_{1}}^{t}He^{-H(s-t_{1})}e^{-\mu
(t-s)}ds+e^{-H(t-t_{1})}+\frac{1-\alpha}{\alpha}e^{2Ht_{1}+H(t-t_{1})}}.  \label{5}
\end{eqnarray}

The numerator of (\ref{5}) is%
\begin{eqnarray}
&&H\{\frac{p_{1}+p_{2}}{2}\int_{t_{1}}^{t}He^{-H(s-t_{1})}e^{-\mu
(t-s)}ds+e^{-H(t-t_{1})}[(p_{1}+p_{2})e^{-H(t_{2}-t)}e^{-\mu (t_{2}-t)} 
\nonumber \\
&&+(p_{1}+p_{2})\int_{t}^{t_{2}}\mu e^{-\mu (s-t)}e^{-H(s-t)}ds+\frac{%
p_{1}+p_{2}}{2}\int_{t}^{t_{2}}He^{-H(s-t)}e^{-\mu (s-t)}ds]\}  \nonumber \\
&=&H\{\frac{p_{1}+p_{2}}{2}\frac{H}{-(H-\mu )}e^{Ht_{1}-\mu t}(e^{-(H-\mu
)t}-e^{-(H-\mu )t_{1}})+e^{-H(t-t_{1})}[(p_{1}+p_{2})e^{-H(t_{2}-t)}e^{-\mu
(t_{2}-t)}  \nonumber \\
&&+(p_{1}+p_{2})\frac{\mu }{-(H+\mu )}e^{Ht+\mu t}(e^{-(H+\mu
)t_{2}}-e^{-(H+\mu )t}) \nonumber \\
&&+\frac{p_{1}+p_{2}}{2}\frac{H}{-(H+\mu )}e^{Ht+\mu
t}(e^{-(H+\mu )t_{2}}-e^{-(H+\mu )t})]\}  \nonumber \\
&=&H\{\frac{p_{1}+p_{2}}{2}(-\frac{H}{H-\mu })(e^{H(t_{1}-t)}-e^{\mu
(t_{1}-t)})+e^{-H(t-t_{1})}[(p_{1}+p_{2})e^{-H(t_{2}-t)}e^{-\mu (t_{2}-t)} 
\nonumber \\
&&+(p_{1}+p_{2})(-\frac{\mu }{H+\mu })(e^{-(H+\mu )(t_{2}-t)}-1)+\frac{%
p_{1}+p_{2}}{2}(-\frac{H}{H+\mu })(e^{-(H+\mu )(t_{2}-t)}-1)]\}  \nonumber \\
&=&H(p_{1}+p_{2})\{(-\frac{\frac{1}{2}H}{H-\mu })(e^{H(t_{1}-t)}-e^{\mu
(t_{1}-t)})+e^{-H(t-t_{1})}[\frac{\frac{1}{2}H}{H+\mu }e^{-(H+\mu
)(t_{2}-t)}+\frac{\frac{1}{2}H+\mu }{H+\mu }]\}  \nonumber \\
&=&H(p_{1}+p_{2})e^{-H(t-t_{1})}[-\frac{\frac{1}{2}H}{H-\mu }+\frac{\frac{1%
}{2}H}{H-\mu }e^{(H-\mu )(t-t_{1})}+\frac{\frac{1}{2}H+\mu }{H+\mu }+\frac{%
\frac{1}{2}H}{H+\mu }e^{-(H+\mu )(t_{2}-t)}].  \label{6}
\end{eqnarray}
Omitting the constant term $H(p_{1}+p_{2})$, we focus on the following
expression: $e^{-H(t-t_{1})}[-\frac{\frac{1}{2}H}{H-\mu }+\frac{\frac{1}{2}H%
}{H-\mu }e^{(H-\mu )(t-t_{1})}+\frac{\frac{1}{2}H+\mu }{H+\mu }+\frac{\frac{1%
}{2}H}{H+\mu }e^{-(H+\mu )(t_{2}-t)}]$. Take the first-order derivative with respect to 
$t$:
\begin{eqnarray}
&&\frac{\partial \{e^{-H(t-t_{1})}[-\frac{\frac{1}{2}H}{H-\mu }+\frac{\frac{1%
}{2}H}{H-\mu }e^{(H-\mu )(t-t_{1})}+\frac{\frac{1}{2}H+\mu }{H+\mu }+\frac{%
\frac{1}{2}H}{H+\mu }e^{-(H+\mu )(t_{2}-t)}]\}}{\partial t}  \nonumber \\
&=&-He^{-H(t-t_{1})}(-\frac{\frac{1}{2}H}{H-\mu }+\frac{\frac{1}{2}H+\mu }{%
H+\mu })-\frac{\frac{1}{2}H\mu }{H-\mu }e^{-\mu (t-t_{1})}+\frac{\frac{1}{2}%
H\mu }{H+\mu }e^{-H(t_{2}-t_{1})-\mu (t_{2}-t)}  \nonumber \\
&=&\frac{\frac{1}{2}H\mu }{H+\mu }[-e^{-H(t-t_{1})}+e^{-H(t_{2}-t_{1})-\mu
(t_{2}-t)}]+\frac{\frac{1}{2}\mu H}{H-\mu }(e^{-H(t-t_{1})}-e^{-\mu
(t-t_{1})}).  \label{7}
\end{eqnarray}%
The first term of (\ref{7}) is smaller than 0 since $%
-H(t-t_{1})>-H(t_{2}-t_{1})-\mu (t_{2}-t)$, also, the second term is smaller
than 0. Hence the numerator of (\ref{5}) is decreasing in $t$.

The denominator of (\ref{5}) is%
\begin{eqnarray*}
&&\int_{t_{1}}^{t}He^{-H(s-t_{1})}e^{-\mu
(t-s)}ds+e^{-H(t-t_{1})}+e^{2Ht_{1}+H(t-t_{1})} \\
&=&\frac{-\mu }{H-\mu }e^{H(t_{1}-t)}+\frac{H}{H-\mu }e^{\mu
(t_{1}-t)}+\frac{1-\alpha}{\alpha}e^{2Ht_{1}+H(t-t_{1})}.
\end{eqnarray*}
Take the first-order derivative with respect to $t$:
\begin{eqnarray}
&&\frac{\partial \lbrack \frac{-\mu }{H-\mu }e^{H(t_{1}-t)}+\frac{H}{H-\mu }%
e^{\mu (t_{1}-t)}+\frac{1-\alpha}{\alpha}e^{2Ht_{1}+H(t-t_{1})}]}{\partial t}  \nonumber \\
&=&\frac{H\mu }{H-\mu }[e^{-H(t-t_{1})}-e^{-\mu
(t-t_{1})}]+\frac{1-\alpha}{\alpha}He^{2Ht_{1}+H(t-t_{1})}.  \label{29}
\end{eqnarray}
The first-order derivative of $e^{-H(t-t_{1})}-e^{-\mu (t-t_{1})}$ with respect to $t$ is%
\[
-He^{-H(t-t_{1})}+\mu e^{-\mu (t-t_{1})},
\]%
which is 0 at%
\begin{eqnarray*}
He^{-H(t-t_{1})} =\mu e^{-\mu (t-t_{1})}.
\end{eqnarray*}
Thus the minimum value of $\frac{H\mu }{H-\mu }[e^{-H(t-t_{1})}-e^{-\mu
(t-t_{1})}]$ is obtained when $He^{-H(t-t_{1})} =\mu e^{-\mu (t-t_{1})}$. Then we have
\begin{eqnarray*}
&&\frac{H\mu }{H-\mu }[e^{-H(t-t_{1})}-e^{-\mu
(t-t_{1})}]+\frac{1-\alpha}{\alpha}He^{2Ht_{1}+H(t-t_{1})} \\
&\geq&\frac{1-\alpha}{\alpha}He^{2Ht_{1}+H(t-t_{1})}+\frac{H\mu }{H-\mu }[e^{-H(t-t_{1})}-e^{-\mu (t-t_{1})}] \\
&=&\frac{1-\alpha}{\alpha}He^{2Ht_{1}+H(t-t_{1})}+\frac{H\mu }{H-\mu }e^{-H(t-t_{1})}[1-\frac{H}{\mu}] \\
&=&\frac{1-\alpha}{\alpha}He^{2Ht_{1}+H(t-t_{1})}-He^{-H(t-t_{1})}\\
&\geq&\frac{1-\alpha}{\alpha}He^{2Ht_{1}}-He^{-H(t-t_{1})}.
\end{eqnarray*}%
The above inequality is binding only when $H=\mu$ and $t=t_1=0$. Thus to prove that the denominator of (\ref{5}) is increasing in $t$ it suffices to prove that $\frac{1-\alpha}{\alpha}e^{2Ht_{1}}\geq 1$. When $\alpha<\frac{1}{2}$, it is clear that $\frac{1-\alpha}{\alpha}e^{2Ht_{1}}\geq 1$ since $t_1\geq 0$. When $\alpha\geq\frac{1}{2}$, from the exit condition (\ref{22}), we have
\begin{align*}
c=&H \frac{(p_{1}+p_{2})e^{-H(t_{2}-t_{1})}+\frac{p_{1}+p_{2}}{2}%
\int_{0}^{t_{2}-t_{1}}He^{-Hs}e^{-(t_{2}-t_{1}-s)\mu }ds}{%
e^{-H(t_{2}-t_{1})}+\int_{0}^{t_{2}-t_{1}}He^{-Hs}e^{-(t_{2}-t_{1}-s)\mu
}ds+\frac{1-\alpha}{\alpha}e^{Ht_{1}+Ht_{2}}}\\
e^{Ht_{1}+Ht_{2}}=&\frac{\alpha}{1-\alpha}(\frac{H(p_{1}+p_{2})[\frac{\frac{1}{2}H-\mu }{H-\mu }%
e^{-H(t_{2}-t_{1})}+\frac{\frac{1}{2}H}{H-\mu }e^{-\mu (t_{2}-t_{1})}]}{c}\\
&+\frac{\mu }{H-\mu }e^{-H(t_{2}-t_{1})}-\frac{H}{H-\mu }e^{-\mu(t_{2}-t_{1})})\\
e^{2Ht_{1}}=&\frac{\alpha}{1-\alpha}(\frac{H(p_{1}+p_{2})[\frac{\frac{1}{2}H-\mu }{H-\mu }%
e^{-H(t_{2}-t_{1})}+\frac{\frac{1}{2}H}{H-\mu }e^{-\mu (t_{2}-t_{1})}]}{c}\\
&+\frac{\mu }{H-\mu }e^{-H(t_{2}-t_{1})}-\frac{H}{H-\mu }e^{-\mu(t_{2}-t_{1})})e^{H(t_1-t_2)}.
\end{align*}
Note that we have proved that $t_2-t_1$ is a constant when $t_2\geq\frac{\ln \frac{H(p_{1}+p_{2})}{p_{1}H-p_{2}\mu }}{H+\mu}$, and that by assumption A.3 we know that $t_2\geq\frac{\ln \frac{H(p_{1}+p_{2})}{p_{1}H-p_{2}\mu }}{H+\mu}$ for every $\alpha\geq\frac{1}{2}$. By letting $\alpha=\frac{1}{2}$ and observing that $t_1\geq 0$, we have
\begin{align*}
(\frac{H(p_{1}+p_{2})[\frac{\frac{1}{2}H-\mu }{H-\mu }%
e^{-H(t_{2}-t_{1})}+\frac{\frac{1}{2}H}{H-\mu }e^{-\mu (t_{2}-t_{1})}]}{c}&\\
+\frac{\mu }{H-\mu }e^{-H(t_{2}-t_{1})}-\frac{H}{H-\mu }e^{-\mu(t_{2}-t_{1})})e^{H(t_1-t_2)}&\geq1\\
e^{2Ht_{1}}&\geq\frac{\alpha}{1-\alpha}\\
\frac{1-\alpha}{\alpha}e^{2Ht_{1}}&\geq 1.
\end{align*}
Therefore the denominator of (\ref{5}) is increasing in $t$. In conclusion, (\ref{5}) is decreasing in $t$ for $t\in \lbrack t_{1},t_{2}]$. Notice that (\ref{5})$=c$ at $t=t_{2}$. Thus the value of staying is always larger than the cost, and the firm will not exit.

\textbf{Part III: prove that withholding the solution after solving stage 1 is optimal in the withhold region.} At time $t\in[t_1, t_2]$, the expected payoff from disclosure is%
\begin{equation}
\frac{(p_{1}+\frac{p_{2}}{2})e^{-H(t-t_{1})}+\frac{p_{1}+p_{2}}{2}%
\int_{t_{1}}^{t}He^{-H(s-t_{1})}e^{-\mu (t-s)}ds}{e^{-H(t-t_{1})}+%
\int_{t_{1}}^{t}He^{-H(s-t_{1})}e^{-\mu (t-s)}ds}.  \label{8}
\end{equation}
The expected payoff from withholding the solution is%
\begin{equation}
\frac{%
\begin{array}{c}
\frac{p_{1}+p_{2}}{2}\int_{t_{1}}^{t}He^{-H(s-t_{1})}e^{-\mu
(t-s)}ds+e^{-H(t-t_{1})}[(p_{1}+p_{2})e^{-H(t_{2}-t)}e^{-\mu (t_{2}-t)} \\ 
+(p_{1}+p_{2})\int_{t}^{t_{2}}\mu e^{-H(s-t)}e^{-\mu (s-t)}ds+\frac{%
p_{1}+p_{2}}{2}\int_{t}^{t_{2}}He^{-H(s-t)}e^{-\mu (s-t)}ds]%
\end{array}%
}{e^{-H(t-t_{1})}+\int_{t_{1}}^{t}He^{-H(s-t_{1})}e^{-\mu (t-s)}ds}.
\label{9}
\end{equation}
Then the difference between withholding and disclosing the solution is (\ref%
{9})-(\ref{8}):%
\begin{eqnarray}
&&\frac{%
\begin{array}{c}
\frac{p_{1}+p_{2}}{2}\int_{t_{1}}^{t}He^{-H(s-t_{1})}e^{-\mu
(t-s)}ds+e^{-H(t-t_{1})}[(p_{1}+p_{2})e^{-H(t_{2}-t)}e^{-\mu (t_{2}-t)} \\ 
+(p_{1}+p_{2})\int_{t}^{t_{2}}\mu e^{-H(s-t)}e^{-\mu (s-t)}ds+\frac{%
p_{1}+p_{2}}{2}\int_{t}^{t_{2}}He^{-H(s-t)}e^{-\mu (s-t)}ds] \\ 
-(p_{1}+\frac{p_{2}}{2})e^{-H(t-t_{1})}-\frac{p_{1}+p_{2}}{2}%
\int_{t_{1}}^{t}He^{-H(s-t_{1})}e^{-\mu (t-s)}ds%
\end{array}%
}{e^{-H(t-t_{1})}+\int_{t_{1}}^{t}He^{-H(s-t_{1})}e^{-\mu (t-s)}ds} 
\nonumber \\
&=&\frac{%
\begin{array}{c}
e^{-H(t-t_{1})}[(p_{1}+p_{2})e^{-H(t_{2}-t)}e^{-\mu
(t_{2}-t)}+(p_{1}+p_{2})\int_{t}^{t_{2}}\mu e^{-H(s-t)}e^{-\mu (s-t)}ds \\ 
+\frac{p_{1}+p_{2}}{2}\int_{t}^{t_{2}}He^{-H(s-t)}e^{-\mu (s-t)}ds-(p_{1}+%
\frac{p_{2}}{2})]%
\end{array}%
}{e^{-H(t-t_{1})}+\int_{t_{1}}^{t}He^{-H(s-t_{1})}e^{-\mu (t-s)}ds} 
\nonumber \\
&=&\frac{e^{-H(t-t_{1})}[(p_{1}+p_{2})e^{-H(t_{2}-t)}e^{-\mu
(t_{2}-t)}+(p_{1}+p_{2})\frac{\mu +\frac{1}{2}H}{\mu +H}(1-e^{-(\mu
+H)(t_{2}-t)})-(p_{1}+\frac{p_{2}}{2})]}{e^{-H(t-t_{1})}+%
\int_{t_{1}}^{t}He^{-H(s-t_{1})}e^{-\mu (t-s)}ds}  \nonumber \\
&=&\frac{e^{-H(t-t_{1})}[\frac{-\frac{1}{2}H}{\mu +H}p_{1}+\frac{\frac{1}{2}%
\mu }{\mu +H}p_{2}+(p_{1}+p_{2})\frac{\frac{1}{2}H}{\mu +H}e^{-(\mu
+H)(t_{2}-t)}]}{e^{-H(t-t_{1})}-\frac{1}{H-\mu }[e^{H(t_{1}-t)}-e^{\mu
(t_{1}-t)}]}  \nonumber \\
&=&\frac{\frac{-\frac{1}{2}H}{\mu +H}p_{1}+\frac{\frac{1}{2}\mu }{\mu +H}%
p_{2}+(p_{1}+p_{2})\frac{\frac{1}{2}H}{\mu +H}e^{-(\mu +H)(t_{2}-t)}}{1-%
\frac{1}{H-\mu }[1-e^{(H-\mu )(t-t_{1})}]}.  \label{10}
\end{eqnarray}

The numerator is increasing in $t$. The denominator is always positive.
Noticeably, by our definition of equilibrium, (\ref{10}) is zero at $t=t_{1}$%
. Then as $t$ increases from $t_{1}$, the numerator is always
positive. Thus $(\ref{10})>0$ for $t\in \lbrack t_{1},t_{2}]$. Hence,
withholding the solution is confirmed to be superior to disclosure
everywhere in the withhold region.

\textbf{Part IV: prove that disclosure is optimal after solving stage 1 in the disclose region.} We discuss two types of deviation: type-1 deviation where the firm will withhold the solution for some period and then disclose the solution, and type-2 deviation where the firm will withhold the solution all the way until exit. We verify that both types of deviation are inferior to immediate disclosure. In the disclose region, i.e., $[0,t_{1}]$, the payoff from disclosure is $p_{1}+\frac{1}{2}p_{2}$.

Suppose that the firm uses type-1 deviation, i.e. it withholds for $[t,t^{\prime }]$ and then discloses the solution. By the proof of Lemma 4, the payoff is always less than $p_{1}+\frac{1}{2}p_{2}$, and thus any type-1 deviation is not profitable.

Suppose that the firm uses type-2 deviation, i.e. it withholds from $t$ until exiting at $t_2$, and suppose that
the expected payoff at time $t_{1}$ is $U$; then, the expected payoff at
time $t$ is%
\begin{eqnarray}
&&e^{-H(t_{1}-t)}e^{-\mu (t_{1}-t)}U+\int_{t}^{t_{1}}\mu e^{-\mu
(s-t)}e^{-H(s-t)}ds (p_{1}+p_{2}) \nonumber \\
&&+\int_{t}^{t_{1}}He^{-H(s-t)}e^{-\mu (s-t)}ds \frac{%
(p_{1}+p_{2})}{2}. \label{16}
\end{eqnarray}
The profit generated by the deviation is (\ref{16}) minus the original
payoff $p_{1}+\frac{1}{2}p_{2}$. However, the firm should be indifferent
between disclosing and withholding the solution at $t_{1}$; thus, $U=p_{1}+\frac{1}{2}%
p_{2} $. Then, following Lemma 4, any type-2 deviation is not profitable. Hence, disclosure is confirmed to be superior to withholding the solution everywhere in the disclose region. This completes the proof.

\section{Proof of Proposition 1}

In a two-firm equilibrium, combine the exit condition and the information
updating process, we have%
\begin{equation}
c=H \frac{(p_{1}+p_{2})e^{-H(t_{2}-t_{1})}+\frac{p_{1}+p_{2}}{2}%
\int_{0}^{t_{2}-t_{1}}He^{-Hs}e^{-(t_{2}-t_{1}-s)\mu }ds}{%
e^{-H(t_{2}-t_{1})}+\int_{0}^{t_{2}-t_{1}}He^{-Hs}e^{-(t_{2}-t_{1}-s)\mu }ds+%
\frac{1-\alpha }{\alpha }e^{Ht_{1}+Ht_{2}}}.  \label{46}
\end{equation}

For one firm without competition, its exit condition is%
\begin{eqnarray}
\frac{e^{-Ht}}{e^{-Ht}+\frac{1-\alpha }{\alpha }}H(p_{1}+p_{2}) &=&c
\label{48} \\
t &=&-\frac{\ln [\frac{c}{H(p_{1}+p_{2})-c}\frac{1-\alpha }{\alpha }]}{H}. 
\nonumber
\end{eqnarray}

We call this exit time $t^{\ast }$ and compare $t^{\ast }$ with $2t_{2}$. Rearranging (\ref{48}) yields%
\begin{eqnarray}
c &=&H\cdot \frac{(p_{1}+p_{2})e^{-2H\frac{t^{\ast }}{2}}}{e^{-2H\frac{%
t^{\ast }}{2}}+\frac{1-\alpha }{\alpha }}  \nonumber \\
&=&H\cdot \frac{(p_{1}+p_{2})e^{-(H\frac{t^{\ast }}{2}-t_{1})}}{e^{-(H\frac{%
t^{\ast }}{2}-t_{1})}+\frac{1-\alpha }{\alpha }e^{Ht_{1}+H\frac{t^{\ast }}{2}%
}}.  \label{50}
\end{eqnarray}%

We substitute $\frac{t^{ast}}{2}$ for $t_2$ in (\ref{46}), and we denote $H(p_{1}+p_{2})e^{-(H\frac{t^{\ast }}{2}-t_{1})}$ as A, $e^{-(H\frac{%
t^{\ast }}{2}-t_{1})}+\frac{1-\alpha }{\alpha }e^{Ht_{1}+H\frac{t^{\ast }}{2}%
}$ as B, and $\int_{0}^{\frac{t^{\ast }}{2}%
-t_{1}}He^{-Hs}e^{-(\frac{t^{\ast }}{2}-t_{1}-s)\mu }ds$ as C. From (\ref{50}%
), $c=\frac{A}{B}$, and we know that%
\[
\frac{A+\frac{H(p_{1}+p_{2})}{2}C}{B+C}\left\{ 
\begin{array}{c}
>c\text{ if and only if }\frac{H(p_{1}+p_{2})}{2}>c; \\ 
=c\text{ if and only if }\frac{H(p_{1}+p_{2})}{2}=c; \\ 
<c\text{ if and only if }\frac{H(p_{1}+p_{2})}{2}<c.%
\end{array}%
\right. 
\]

From the proof of Theorem 1, the RHS of (\ref{46}) is decreasing in $t_{2}$.
Thus, in equilibrium%
\[
t_{2}\left\{ 
\begin{array}{c}
>\frac{t^{\ast }}{2}\text{ if and only if }\frac{H(p_{1}+p_{2})}{2}>c; \\ 
=\frac{t^{\ast }}{2}\text{ if and only if }\frac{H(p_{1}+p_{2})}{2}=c; \\ 
<\frac{t^{\ast }}{2}\text{ if and only if }\frac{H(p_{1}+p_{2})}{2}<c.%
\end{array}%
\right. 
\]

\end{document}